\documentclass{aptpub}
\usepackage{amsmath}
\usepackage[dvips]{graphicx}
\usepackage{amssymb, amsmath, amsfonts, latexsym}
\usepackage{color}
\usepackage{amssymb}
\newcommand{\no}{\nonumber}

\authornames{Amogh Deshpande} 
\shorttitle{On the role of F\"{o}llmer-Schweizer minimal martingale measure in  Risk Sensitive control Asset Management.} 

\begin{document}
\textcolor{red}{A.Deshpande (2015), On the role of F\"{o}llmer-Schweizer minimal martingale measure in  Risk Sensitive control Asset Management,Vol. 52, No. 3, Journal of Applied Probability.}

\title{On the role of F\"{o}llmer-Schweizer minimal martingale measure in  Risk Sensitive control Asset Management} 
\authorone[University of Warwick,~UK]{Amogh Deshpande} 

\addressone{ Department of Statistics, University of Warwick, UK, CV47AL. Email: addeshpa@gmail.com} 

\begin{abstract}

Kuroda and Nagai \cite{KN} state that the factor process in the Risk Sensitive control Asset Management (RSCAM) is stable under the F\"{o}llmer-Schweizer minimal martingale measure . Fleming and Sheu \cite{FS} and more recently F\"{o}llmer and Schweizer \cite{FoS} have observed that  the role of the minimal martingale measure in this portfolio optimization is yet to be established. In this article we aim to address this question by explicitly connecting the optimal wealth allocation to the minimal martingale measure. We achieve this by using a ``trick" of observing this problem in the context of model uncertainty via a two person zero sum stochastic differential game between the investor and an antagonistic market that provides a probability measure. We obtain some startling insights. Firstly,  if short-selling is not permitted and if the factor process evolves under the minimal martingale measure  then the investor's optimal strategy can only be to invest in the riskless asset (i.e. the no-regret strategy). Secondly, if the factor process and the stock price process have independent noise, then even if the market allows short selling, the optimal strategy for the investor must be the no-regret strategy while the factor process will evolve under the minimal martingale measure .
\end{abstract}

\keywords{Risk Sensitive control Asset Management;  Minimal martingale measure; zero sum stochastic differential game;  stability.} 
\ams{49L02}{60G02} 

\section{Introduction} 
$\indent$ Risk sensitive control Asset Management  (RSCAM) balances the investor�s interest in maximizing the expected growth rate of wealth against his
aversion to risk due to deviations of the actually realized rate from the expectation for a finite time horizon. The subjective notion of investor�s risk aversion is parameterized by a single variable say $\theta$. In RSCAM we consider the following criterion to be maximized. For a given $\theta>-2,\theta\neq 0$  and for time horizon $T<\infty$, find wealth allocation control denoted by $h(t)$ ,
 the risk- sensitive expected growth rate up to time horizon $T$ criterion $J(v,h,T;\theta)$ defined by,
\begin{eqnarray}\label{1.1}
J(v,h,T;\theta)\triangleq \frac{-2}{\theta}\log E[\exp{[\frac{-\theta}{2}\log{V^{h}({T})}]}]
\end{eqnarray}
where $V^{h}({T})$ is time-$T$ portfolio value.
An asymptotic expansion around $\theta=0$ for the above criterion yields
\begin{eqnarray*}
J(v,h,T;\theta)=E[V^{h}({T})]-\frac{\theta}{2}Var(V^{h}({T}))+O(\theta^{2});~~~~V^{h}(0)=v
\end{eqnarray*}
As is obvious from the preceding equation, $\theta>0$ corresponds to risk averse investor, $\theta<0$ is risk seeking investor and $\theta=0$ is a risk neutral investor.  Hence the optimal expected utility function depends on $\theta$ and is a generalization of the traditional  stochastic control in the sense that now the degree of risk aversion  of the investor is explicitly parameterized through $\theta$ rather than importing it in the problem via an externally defined utility function.   For this reason investment optimization models have been popularly reformulated as  risk-sensitive control problems . For a general reference on risk-sensitive control, refer Whittle \cite{W}. \\
\indent Risk-sensitive control was first applied to solve financial problems by Lefebvre and Montulet \cite{LM} in a corporate finance context and by Fleming \cite{Fleming} in a portfolio selection context. A RSCAM problem  with $m$ securities and $n$ (economic) factors was introduced by Bielecki and Pliska \cite{BP}. Their factor model however made a rather strong assumption that the factor process and the securities price process in their financial optimization model had independent noise. A generalization to this model, relaxing this assumption was made by Kuroda and Nagai \cite{KN} who introduced an elegant solution method based on a change of measure argument which transforms the risk sensitive control problem into a linear exponential of a quadratic regulator. They  solved the associated HJB PDE over a finite time horizon and then studied the properties of the ergodic HJB PDE. We go about formally stating the problem by first describing the factor model for a risk averse investor.\\
\indent Let ($\Omega,\mathcal{F},(\mathcal{F}_{t})_{t\geq 0},\mathbb{P}$) be the filtered probability space. Consider a market of $m+1 \geq 2$ securities and $n\geq 1$ factors. We assume that the set of securities includes one bond whose price is governed  by the ODE
\begin{eqnarray}\label{1.2}
dS^{0}({t})=r({t})S^{0}({t})dt, ~~~~S^{0}({0})=s^{0}
\end{eqnarray}
where $r({t})$ is a deterministic function of $t$. The other security prices are assumed to satisfy the following SDE's
\begin{eqnarray}\label{1.3}
dS^{i}({t})=S^{i}({t})\{(a+AX({t}))^{i}dt+\sum_{k=1}^{n+m}{\sigma_{k}^{i}dW^{k}({t})}\}, S^{i}({0})=s^{i}, i=1,...,m.,
\end{eqnarray}
where the component wise factor process satisfies,
\begin{eqnarray*}
dX^{i}({t})&=&(b+BX({t}))^{i}dt+\sum_{k=1}^{n+m}{\lambda_{k}^{i}} dW^{k}({t}),~~~ X^{i}({0})=x^{i} , i=1,...,n.
\end{eqnarray*}
Vectorically $X(t)=(X^{1}(t),...,X^{n}(t))^{'}$ (where the symbol $'$ signifies transpose)  satisfies the following dynamics,
\begin{eqnarray}\label{1.4}
dX({t})&=&(b+BX({t}))dt+\Lambda dW(t), X(0)=x \in \mathbb{R}^{n}.
\end{eqnarray}
Here,
$W({t})=(W^{k}({t}))_{k=1,...,n+m}$ is an $n+m$ dimensional standard Brownian motion  defined on the filtered probability space. The model parameters $A, B$ are respectively $m \times n, n \times n, n \times (m+n)$ constant matrices and $a \in \mathbb{R}^{m}$ , $b \in \mathbb{R}^{n}$. The constant matrix $[\sigma_{k}^{i}] \triangleq \Sigma, i=1,2....,m; k=1,2,...,(n+m)$. Matrix $\Sigma\Sigma^{'}$ is assumed positive definite. Similarly, $[\lambda_{k}^{i}]\triangleq \Lambda, i=1,2....,n; k=1,2,...,(n+m)$.
We denote $l^{'}$ as transpose of $l$. Likewise let $|v |$  be a suitable vector norm for any vector $v$ while $||M||$ symbolizes a suitable matrix norm for any matrix $M$. As discussed earlier, as part of generalizing the Bielecki and Pliska factor model \cite{BP}, Kuroda and Nagai \cite{KN} assume that the factor process and the securities price process is correlated i.e. $\Sigma\Lambda^{'}\neq 0$. The  investment strategy which represents proportional allocation of total wealth in the $i^{th}$ security $S^{i}({t})$ is denoted by $h^{i}({t})$  for $i=0,1,...,m$ and we set, $S({t}):=(S^{1}({t}),S^{2}({t}),...,S^{m}({t}))^{'},  h({t}):=(h^{1}({t}),...,h^{m}({t}))^{'}$ and $\mathcal{G}_{t}=\sigma(S({u}),X({u});u \leq t)$ is the filtration generated by the underlying stock price process and the factor process. Let $\mathcal{H}(T)$ be a space of $\mathbb{R}^{m}$ valued controls for the investor meaning we say that $h({t})\in \mathcal{H}(T)$ where $h(t)$   is $\mathcal{G}_{t}$-progressively measurable stochastic processes such that $\sum_{i=1}^{m}{h^{i}({t})}+h^{0}({t})=1$,
$P(\int_{0}^{T}{{|h({t})|}^{2}dt} < \infty)=1$ and $E[e^{\frac{\theta^{2}}{2}\int_{0}^{T}{h^{'}_{t}\Sigma\Sigma^{'}h_{t}}dt}]<\infty$. For given $h({t})\in \mathcal{H}(T)$ the process $V({t})=V^{h}({t})$  is determined by the SDE,
\begin{eqnarray*}
\frac{dV^{h}({t})}{V^{h}({t})}=h^{0}({t})r({t})dt + \sum_{i=1}^{m}{h^{i}({t})\{(a+AX({t}))^{i}dt}+\sum_{k=1}^{m+n}{\sigma_{k}^{i}dW^{k}({t})}\}; ~~~~V^{h}({0})=v.
\end{eqnarray*}
which can be written vectorically as,
\begin{eqnarray}\label{1.5}
\frac{dV^{h}{(t)}}{V^{h}{(t)}}=(r({t}) + h^{'}(t)\delta{(t)})dt+ h^{'}{(t)}\Sigma dW{(t)};~~~ V^{h}{(0)}=v.
\end{eqnarray}
where $\delta{(t)}\triangleq a+AX{(t)}-r{(t)}1$. From the expression of security/stock price dynamics $S({t})$ (\ref{1.3}), it is obvious that the market is incomplete (as it has $m $ securities and $n+m$ Brownian drivers) and hence there exist many equivalent martingale measures or EMM's. We refer the reader to Karatzas and Shreve \cite{KS} for  a general treatment on market incompleteness. One such candidate equivalent martingale measure   is the F\"{o}llmer-Schweizer minimal martingale measure. For the continuous adapted stock price process $S=(S(t))_{0 \leq t \leq T}$, the minimal martingale measure $\mathbb{P}^{*}$ (say)is the unique equivalent local martingale measure with the property that local $\mathbb{P}$-martingale part of $S$ are also local $\mathbb{P}^{*}$-martingales.
For the F\"{o}llmer-Schweizer minimal martingale measure $\mathbb{P}^{*}$, the density process is given by the following dynamics, \\
\begin{eqnarray}\label{1.6}
\frac{d\mathbb{P}^{*}}{d\mathbb{P}}=\mathcal{E}(-\int_{0}({(\Sigma^{'}(\Sigma\Sigma^{'})^{-1})~ \delta})^{'}~dW)_{T}
\end{eqnarray}
\indent  Kuroda and Nagai \cite{KN} observe that the condition of stability of  the matrix $B-\Lambda\Sigma^{'} (\Sigma\Sigma^{'})^{-1}A$ induces stability on the  factor process $X=(X(t))_{0 \leq t \leq T}$  under the minimal martingale measure. Fleming and Sheu \cite{FS} and more recently F\"{o}llmer and Schweizer \cite{FoS} have observed that this observation and more significantly the role of the minimal martingale measure in this portfolio asset management problem is yet to be established. In this article we address these questions. We do so by conceptualizing the RSCAM as a zero sum stochastic differential game between (a market) that provides a probability measure that works antagonistically against another player (the investor) who otherwise wants to maximize the risk-sensitive criterion.  We call this game ({\bf GI})(refer (\ref{2.4})). We need to determine the  controls that forms the saddle point equilibrium to this game. This will then illuminate the explicit dependence between controls $h(t)$ and the probability measure which would then lead us to connect the role played by the minimal martingale measure. We achieve this objective through the following road map:\\
{\bf Key Steps}\\
{\bf Step 1}: We  re-formulate the game ({\bf GI}) into an auxiliary game characterized by the exponential of integral criterion that involves just the factor process $X$. We call this game ({\bf GII})(refer equation (\ref{2.11})).\\
{\bf Step 2}: We then provide verification lemma for ({\bf GII}).\\
{\bf Step 3}: We then obtain the optimal controls and deduce the connection between the minimal martingale measure and investor's optimal strategy.  \\
{\bf Step 4}: To complete the analysis we end by showing that the controls hence obtained while solving game ({\bf GII}) in Step 3 infact also constitutes saddle-point equilibrium strategy for the original game ({\bf GI}).\\
\section{Worst-Case Risk sensitive Zero sum  stochastic differential game}
As discussed in the introduction, the Kuroda and Nagai investment market model is incomplete. We are interested in understanding the influence minimal martingale measure has on this portfolio optimization problem. We conjure an approach, whereby we can  explicitly characterize the dependence between the minimal martingale measure and the control variable $h$. Formally we set to do this is to define a ``market world". The market world is a space of  probability measures defined as
 \begin{eqnarray*}
 \mathcal{P} \triangleq \{\mathbb{P}^{\eta,\xi}: (\eta,\xi)=(\eta({t}),\xi({t}))_{T \geq t\geq 0} \in \mathcal{O}(T)\}
 \end{eqnarray*}
 on ($\Omega ,\mathcal{F}$), where $\mathcal{O}(T)$ denotes the set of deterministic controls $\eta(t) \in \mathbb{R}^{n \times (n+m)}$ and $\xi(t)\in \mathbb{R}^{1 \times (n+m)}$ which are continuous over the compact set $[0,T]$ and hence bounded. For $(\eta(t),\xi(t)) \in \mathcal{O}(T)$ for fixed time horizon $T$, the restriction of $\mathbb{P}^{\eta,\xi}$ to the $\sigma-$ field $\mathcal{F}_{T}$ is given by the Radon-Nikodym density
 \begin{eqnarray}\label{2.1}
D^{\eta,\xi}({T}) \triangleq \frac{d \mathbb{P}^{\eta,\xi}}{d\mathbb{P}}|_{\mathcal{F}_{T}} \triangleq \mathcal{E}\bigg(\int_{0}{(\eta({t})^{'}X(t)+\xi^{'}(t))^{'}}dW(t)\bigg)_{T}.
 \end{eqnarray}
 with respect to the reference measure $\mathbb{P}$.  Here $\mathcal{E}(\cdot)$ is the Dole\'ans-Dade exponential. We now show that for $(\eta,\xi) \in \mathcal{O}(T)$, $\mathbb{P}^{\eta,\xi}$ is a probability measure.\\
\begin{lem}{$E[D^{\eta,\xi}(T)]=1$ for all $(\eta, \xi) \in \mathcal{O}(T)$.} \end{lem}
\begin{proof}
 The process  $X(t)$ in (\ref{1.4}) is a Gaussian process. From (\ref{1.4}) and the Gronwall's inequality we have $E|X(t)|\leq (E|X(0)|+|b|T)\exp(||B||t)$ and $Cov(X(t))=\Lambda^{'}\Lambda t$ where $Cov$ is the covariance function. As $\eta(t), \xi(t)$ are deterministic controls and are bounded, $\phi(t) \triangleq X^{'}(t)\eta(t)+\xi(t)$ is also a  Gaussian process with bounded mean and covariance on a finite time interval $[0,T]$. Hence by an application of Novikov's condition,  the Dole\'ans-Dade exponential in (\ref{2.1}) is a $\mathbb{P}$- martingale. A standard proof of this fact can be seen in  Lemma 3.1.1 in Bensoussan \cite{Ben}.
 \end{proof}
  We now re-evaluate the optimization criterion $J$ under the new probability measure $\mathbb{P}^{\eta,\xi}$ and call it $\tilde{J}$ which is defined as
\begin{eqnarray*}
\tilde{J}(v,h,\eta,\xi,T;\theta)= \frac{-2}{\theta}\log E^{\eta,\xi}[\exp{[\frac{-\theta}{2}\log{V^{h,\eta,\xi}({T})}]}].
\end{eqnarray*}
where the portfolio value under the new  probability measure $\mathbb{P}^{\eta,\xi}$ is given by
\begin{eqnarray}\label{2.2}
\frac{dV^{h,\eta,\xi}({t})}{V^{h,\eta,\xi}({t})}&=&\bigg[r(t)+ h^{'}({t})(\delta({t})-\Sigma(\eta^{'}(t)X(t)+\xi^{'}(t)))\bigg]dt+ h^{'}({t})\Sigma dW^{\eta,\xi}({t}) , \no\\
 V^{h,\eta,\xi}({0})&=&v.
\end{eqnarray}
From Lemma 2.1 we have that $\mathbb{P}^{\eta,\xi}$ is a probability measure for $(\eta, \xi) \in \mathcal{O}(T)$. \\
From the standard result in Girsanov \cite{Girsanov}, under the probability measure $\mathbb{P}^{\eta,\xi}$,
 \begin{eqnarray*}
W^{\eta,\xi}(t)\triangleq W(t)+\int_{0}^{t}{(\eta^{'}(s)X(s)+\xi^{'}(s))}ds,
 \end{eqnarray*}
 is a standard Brownian motion process and therefore the factor process $X(t)$, vectorically, satisfies the following SDE
 \begin{eqnarray}\label{2.3}
dX(t)=(b+BX(t)-\Lambda (\eta^{'}(t)X(t)+\xi^{'}(t)))dt+\Lambda dW^{\eta,\xi}(t),
 \end{eqnarray}
 \begin{rem}
 From equations (\ref{1.6}) and (\ref{2.1}) , it is clear that $\mathbb{P}^{\eta,\xi}$ is a minimal martingale measure for $\hat{\eta}(t)\triangleq \eta(t)=A^{'}(\Sigma\Sigma^{'})^{-1}\Sigma$ and $\hat{\xi}(t) \triangleq \xi(t)=(a-r(t)1)^{'}(\Sigma\Sigma^{'})^{-1}\Sigma$. \\
\indent Kuroda and Nagai \cite{KN} have stated that  under the condition of stability of the matrix  $B-\Lambda\Sigma^{'} (\Sigma\Sigma^{'})^{-1}A$, the factor process $X(t)$ is stable under the minimal martingale measure. In light of our Remark 2.1, we validate this statement now.
\end{rem}
\begin{rem}
As $\eta(t)=\hat{\eta}(t)$ and $\xi(t)=\hat{\xi}(t)$ corresponds to the minimal martingale measure, the dynamics of $X(t)$ under the minimal martingale measure can be re-written as
\begin{eqnarray*}
dX(t)=\bigg(b-\Lambda \Sigma^{'}(\Sigma\Sigma^{'})^{-1}(a-r(t)1)+(B-\Lambda\Sigma^{'}(\Sigma\Sigma^{'})^{-1}A)X(t)\bigg)dt+\Lambda dW^{\hat{\eta},\hat{\xi}}(t).
\end{eqnarray*}
 We are interested in finding the behavior of the solution $X(t)$ as $t \rightarrow \infty$. The coefficient of the $X(t)$ term in the drift part of above equation  is  $B-\Lambda\Sigma^{'}(\Sigma\Sigma^{'})^{-1}A$. Since by assumption this coefficient term is a stable matrix, $X(t)$ is hence stable under  the minimal martingale measure.
 \end{rem}
\indent  We need to now pin down the influence  the minimal martingale measure has on this portfolio optimization problem to further resolve the inquiry posed by Fleming and Sheu \cite{FS}. \\
To do so, as stated earlier, we conceptualize this problem as a game  between a player termed as  the market against the investor. We denote this game as  ({\bf GI}).\\
{Game {\bf GI}}\textit{
Obtain $\hat{h} \in \mathcal{H}(T)$ and $(\hat{\eta},\hat{\xi})\in \mathcal{O}(T)$ such that,
\begin{eqnarray}\label{2.4}
{\tilde{J}(v,\hat{h},\hat{\eta},\hat{\xi},T;\theta)}
&=&\sup_{h \in {\mathcal{H}}(T)}\inf_{(\eta,\xi) \in \mathcal{O}(T)}{\frac{-2}{\theta}\log E^{\eta,\xi}[\exp{[\frac{-\theta}{2}\log{V^{h,\eta,\xi}({T})}]}]}\no\\
&=&\inf_{(\eta,\xi) \in \mathcal{O}(T)}\sup_{h \in {\mathcal{H}}(T)}{\frac{-2}{\theta}\log E^{\eta,\xi}[\exp{[\frac{-\theta}{2}\log{V^{h,\eta,\xi}({T})}]}]}.
\end{eqnarray}}
Our intention is to re-write the objective function $\tilde{J}$  purely in terms of the factor process $X$. We set to achieve this by defining,
\begin{eqnarray}\label{2.5}
g(x,h,\eta,\xi,r;\theta) &\triangleq & \frac{1}{2}(\frac{\theta}{2}+1)h^{'}\Sigma\Sigma^{'}h-r-h^{'}\bigg(\delta- \Sigma(\eta^{'}x+\xi^{'})\bigg).\no\\
\end{eqnarray}
Hence from (\ref{2.5}), we have
\begin{eqnarray}\label{2.6}
-\frac{\theta}{2}d{\log{V^{h,\eta,\xi}}(t)}=\bigg(\frac{\theta}{2}g(X(t),h(t),\eta(t),\xi(t),r(t);\theta)-\frac{\theta^{2}}{8}h^{'}(t)\Sigma\Sigma^{'}h(t)\bigg)dt-\frac{\theta}{2}h^{'}(t)\Sigma dW^{\eta,\xi}(t).
\end{eqnarray}
We next define the following stochastic exponential given as,
\begin{eqnarray}\label{2.7}
\frac{d \mathbb{P}^{h,\eta,\xi}}{d\mathbb{P}^{\eta,\xi}}|_{\mathcal{F}_{T}} = \mathcal{E}(-\frac{\theta}{2}\int_{0} h^{'}(t)\Sigma dW^{\eta,\xi}(t))_{T}.
\end{eqnarray}
From the definition of the class of controls $\mathcal{H}(T)$, it is clear from an application of Novikov's condition that $  \mathbb{P}^{h,\eta,\xi}$ is a probability measure. Under this probability measure $\mathbb{P}^{h,\eta,\xi}$, the standard result of Girsanov \cite{Girsanov} yields that
\begin{eqnarray*}
W^{h,\eta,\xi}(t) \triangleq W^{\eta,\xi}(t)+\int_{0}^{t}{\frac{\theta}{2}\Sigma^{'}h(s)}ds,
\end{eqnarray*}
is a standard $\mathbb{P}^{h,\eta,\xi}$- Brownian motion and the factor process $X(t)$ satisfies the following dynamics
\begin{eqnarray}\label{2.8}
dX(t)=(b+BX(t)-\Lambda (\eta^{'}(t)X(t)+\xi^{'}(t))-\frac{\theta}{2}\Lambda\Sigma^{'}h(t))dt+\Lambda dW^{h,\eta,\xi}(t).
\end{eqnarray}
Now, under the new probability measure $\mathbb{P}^{h,\eta,\xi}$, and using (\ref{2.4})-(\ref{2.6}) and (\ref{2.8})  we define an auxiliary optimization  criterion $I(v,x,h,\eta,\xi,t,T;\theta)$ given as
\begin{eqnarray}\label{2.9}
I(v,x,h,\eta,\xi,t,T;\theta)=\log{v}-\frac{2}{\theta}\log E^{h,\eta,\xi}\bigg[\exp\bigg(\frac{\theta}{2}\int_{0}^{T-t}{g(X(s),h(s),\eta(s),\xi(s),r(s+t);\theta)}ds\bigg)\bigg].
\end{eqnarray}
This will lead us to frame the auxiliary game {\bf GII} that constitutes our first step as stated under the road map in the Introduction.\\
{\bf \textit{\underline{Step 1:}}}\\
In a worst-case risk-sensitive asset management scenario, the investor chooses a portfolio process $h$ so as to maximize the expected exponential-of-integral  performance index $I$. Then the response of the market  to this choice is to select ($\eta,\xi$) (and hence a probability measure)that minimizes the maximum expected exponential-of-integral  performance index. Formally, \\
{The upper value of this game is given by}\\
\begin{eqnarray*}
\bar{u}(t,x)=\sup_{h \in {\mathcal{H}(T)}}\inf_{(\eta,\xi)\in {\mathcal{O}(T)}}I(v,x,h,\eta,\xi,t,T;\theta),
\end{eqnarray*}
{while the lower value of the game is given by }\\
\begin{eqnarray*}
\underline{u}(t,x)=\inf_{(\eta,\xi)\in {\mathcal{O}(T)}}\sup_{h \in {\mathcal{H}(T)}}I(v,x,h,\eta,\xi,t,T;\theta),
\end{eqnarray*}
The game has a value provided,
\begin{eqnarray}\label{2.10}
\bar{u}(t,x)=\underline{u}(t,x)=u(t,x)=I(v,x,\hat{h},\hat{\eta},\hat{\xi},t,T;\theta).
\end{eqnarray}
and hence $ \hat{h}, (\hat{\eta}, \hat{\xi})$ is a saddle-point equilibrium. We aim to provide a verification lemma for which (\ref{2.10}) is satisfied. In that spirit,  consider the \textit{exponentially transformed} criterion which is simply obtained via the transformation ${\tilde{u}}(t,x)=\exp(-\frac{\theta}{2}{u}(t,x))$. This transformation defines what we call as game {\bf GII}.\\
{Game ({\bf GII})}\\
Obtain $\hat{h}\in \mathcal{H}(T)$ and $(\hat{\eta},\hat{\xi}) \in {\mathcal{O}}(T)$ such that,
\begin{eqnarray}\label{2.11}
\tilde{u}(t,x)&=&\inf_{h \in {\mathcal{H}}(T)}\sup_{(\eta,\xi)\in {\mathcal{O}(T)}}E^{{h},{\eta},{\xi}}[\exp\{\frac{\theta}{2}\int_{0}^{T-t}{g(X({s}),{h}({s}),{\eta}({s}),{\xi}({s}),r(s+t);\theta)}ds\}v^{-\theta/2}],\no\\
&=&\sup_{(\eta,\xi)\in {\mathcal{O}(T)}}\inf_{h \in {\mathcal{H}}(T)}E^{{h},{\eta},{\xi}}[\exp\{\frac{\theta}{2}\int_{0}^{T-t}{g(X({s}),{h}({s}),{\eta}({s}),{\xi}({s}),r(s+t);\theta)}ds\}v^{-\theta/2}],\no\\
&=& E^{\hat{h},\hat{\eta},\hat{\xi}}[\exp\{\frac{\theta}{2}\int_{0}^{T-t}{g(X({s}),\hat{h}({s}),\hat{\eta}({s}),\hat{\xi}({s}),r(s+t);\theta)}ds\}v^{-\theta/2}].
\end{eqnarray}
\section{ An HJBI equation for game GII.}
{\bf \textit{\underline{Step 2:}}}\\
Let us now define a couplet process $Y^{h,(\eta,\xi)}(t) $ as \\
$dY^{h,(\eta,\xi)}(s)= \begin{pmatrix} dY_{0}(s) \\ dY_{1}(s) \end{pmatrix}=\begin{pmatrix} ds \\ dX(s) \end{pmatrix}=\begin{pmatrix} ds \\ 
(b+BX(s)-\Lambda (\eta^{'}(s)X(s)+\xi^{'}(s))-\frac{\theta}{2}\Lambda\Sigma^{'}h(s))dt+\Lambda dW^{h,\eta,\xi}({s}) \end{pmatrix}$\\
$Y_{0}(0)=s \in [0,T], Y_{1}(0)=y=(y^{1},...y^{n})$. The control process $h(s)=h(s,\omega)$ is assumed to be Markovian.
Then the process $Y^{h,(\eta,\xi)}(s)$  is a  Markov  process whose generator acting on a function $\tilde{u}(y) \in C^{1,2}_{0}((0,T) \times \mathbb{R}^{n})$ where ($C^{1,2}_{0}$ is the space of functions with compact support on $(0,T) \times \mathbb{R}^{n}$ such that it is once continuously differentiable in time and twice continuously differentiable in space variable $x$) is given by,
\begin{eqnarray}\label{3.1}
\tilde{\mathcal{A}}^{h,(\eta,\xi)}\tilde{{u}}(y) &=&\frac{\partial \tilde{u}(y)}{\partial s}+(b+Bx-\Lambda (\eta^{'}x+\xi^{'})-\frac{\theta}{2}\Lambda\Sigma^{'} h)^{'}D{\tilde{u}}(y) +\frac{1}{2}tr(\Lambda\Lambda^{*}D^{2}{\tilde{u}}(y)).
\end{eqnarray}
in which $D\tilde{u}(y) \triangleq (\frac{\partial \tilde{u}(y)}{\partial y_{1}^{1}},...,\frac{\partial \tilde{u}(y)}{\partial y_{1}^{n}})^{'}$ and $D^{2}\tilde{u}(y)$ is the matrix defined as $D^{2}\tilde{u}(y)\triangleq [\frac{\partial^{2}\tilde{u}(y)}{\partial y_{1}^{i}\partial y_{1}^{j}}],i,j=1,2,...,n.$\\
\indent By an application of the Feynman-Kac formula, it can be deduced that the HJB PDE for $\tilde{u}(y)$ is given by
\begin{eqnarray}\label{3.2}
\bigg(\tilde{\mathcal{A}}^{\hat{h},(\hat{\eta},\hat{\xi})}+\frac{\theta}{2}g(x,\hat{h}(y),\hat{\eta},\hat{\xi},r;\theta)\bigg){\tilde{u}}(y)=0.
\end{eqnarray}
The following  proposition presents a diagnostic to identify a  solution to the game ({\bf {GII}}).\\
\\
\begin{prop}
Define $\mathcal{S}=(0,T)\times \mathbb{R}^{n}$. Let  there exists a  function $\tilde{w}$ $\in$ $\mathcal{C}^{1,2}({\mathcal{S}}) \cap \mathcal{C}(\bar{\mathcal{S}}) $. Suppose there exists (Markov) control $\hat{h} \in {\mathcal{H}(T)}$  and deterministic controls $(\hat{\eta}, \hat{\xi})\in {\mathcal{O}}(T)$ such that for each $y \in \mathcal{S}$,\\
1. $({\mathcal{\tilde{A}}}^{{h},(\hat{\eta},\hat{\xi})}+\frac{\theta}{2}g(x,h,\hat{\eta},\hat{\xi},r;\theta))[(\tilde{w}(y))] \geq 0~ \forall~ h \in \mathbb{R}^{m}$;\\
2. $({\mathcal{\tilde{A}}}^{\hat{h}(y),({\eta,\xi})}+\frac{\theta}{2}g(x,\hat{h}(y),\eta,\xi,r;\theta))[(\tilde{w}(y))]  \leq 0~ \forall~\eta \in  \mathbb{R}^{n \times (n+m)}, \xi \in \mathbb{R}^{1 \times (n+m)} $;\\
3. $({\mathcal{\tilde{A}}}^{\hat{h}(y),(\hat{\eta},\hat{\xi})}+\frac{\theta}{2}g(x,\hat{h}(y),\hat{\eta},\hat{\xi},r;\theta))[(\tilde{w}(y))]  = 0$;\\
4. $[(\tilde{w}(T,X_{T}))]={v}^{-\theta/2}$.\\
5. $E^{{h},{\eta,\xi}}[\int_{0}^{T-t}{D\tilde{w}^{'}(t+s,X({s}))\Lambda}e^{\tilde{Z}_{s}}dW^{h,\eta,\xi}_{s}]=0 ~ \forall~ h \in \mathbb{R}^{m},\forall~\eta \in  \mathbb{R}^{n \times (n+m)}, \xi \in \mathbb{R}^{1 \times (n+m)} $;\\
where,
\begin{eqnarray}\label{3.3}
\tilde{Z}({s})=\tilde{Z}_{s}(h,\eta,\xi):=\frac{\theta}{2}\bigg\{\int_{0}^{s}{g(X(\tau),h(\tau),\eta(\tau),\xi(\tau),r({t+\tau});\theta)}d{\tau}\bigg\}.
\end{eqnarray}
Define ,
\begin{eqnarray*}
\tilde{I}(v,x,h,\eta,\xi,t,T;\theta)&=&\exp(-\frac{\theta}{2}I(v,x,h,\eta,\xi,t,T;\theta))\\
&=& E^{h,\eta,\xi}[\exp\{\frac{\theta}{2}\int_{0}^{T-t}{g(X({s}),h({s}),\eta({s}),\xi({s}),r({s+t});\theta)}ds\}v^{-\theta/2}].
\end{eqnarray*}
then,
\begin{eqnarray*}
\tilde{u}(0,x)=\tilde{w}(0,x)=\tilde{I}({v,x,\hat{h},\hat{\eta},\hat{\xi},0,T};\theta)& = & \inf_{h \in {\mathcal{H}}(T)}\{\sup_{(\eta,\xi) \in \mathcal{O}(T)}[\tilde{I}({v,x,{h},{\eta,\xi},0,T};\theta)]\},  \\
&=& \sup_{(\eta,\xi) \in \mathcal{O}(T)}\{\inf_{h \in {\mathcal{H}}(T)}[\tilde{I}({v,x,{h},{\eta,\xi},0,T};\theta)]\},  \\
&=& \sup_{(\eta,\xi) \in \mathcal{O}(T)}\tilde{I}({v,x,\hat{h},{(\eta,\xi)},0,T};\theta),  \\
&=& \inf_{h \in {\mathcal{H}}(T)}\tilde{I}({v,x,{h},\hat{\eta},\hat{\xi},0,T};\theta)=\tilde{I}({v,x,\hat{h},\hat{\eta},\hat{\xi},0,T};\theta).
\end{eqnarray*}
and ($\hat{h},(\hat{\eta},\hat{\xi})$) is a saddle point equilibrium.
\end{prop}
\begin{proof}
Apply  Ito's formula to $\tilde{w}(s,X({s}))e^{\tilde{Z}({s})}$ to obtain
\begin{eqnarray}\label{3.4}
\tilde{w}(T,X({T-t}))e^{\tilde{Z}{(T-t)}}&=&\tilde{w}(t,x)\no\\
&+&\int_{0}^{T-t}{((\tilde{\mathcal{A}}^{h,\eta,\xi}+\frac{\theta}{2}g(X({s}),h(X(s)),\eta(s),\xi(s),r({s+t});\theta))\tilde{w}(t+s,X(s)))e^{\tilde{Z}_{s}})}ds\no\\
&+&\int_{0}^{T-t}{(D\tilde{w}^{'}(t+s,X({s}))\Lambda)e^{\tilde{Z}({s})}}dW^{h,\eta,\xi}({s}).
\end{eqnarray}
Taking expectation with respect to $\mathbb{P}^{h,\eta,\xi}$ ,  from condition (5) the Proposition  , the stochastic integral in (\ref{3.4}) vanishes.  Now setting $t=0$ and further applying  condition (1) and (4) again of the Proposition , we get
\begin{eqnarray*}
E^{{h},\eta,\xi}[\tilde{w}(T,X_{T})e^{\tilde{Z}_{T}}]\geq \tilde{w}(0,x).
\end{eqnarray*}
Since this inequality is true for all $h \in {\mathcal{H}}(T)$ we have
\begin{eqnarray*}
\inf_{h \in {\mathcal{H}}(T)}E^{{h},{\eta,\xi}}[v^{-\theta/2}e^{\tilde{Z}_{T}}]\geq \tilde{w}(0,x).
\end{eqnarray*}
Hence we have,
\begin{eqnarray}\label{3.5}
\sup_{(\eta,\xi) \in \mathcal{O}(T)}\inf_{h \in {\mathcal{H}}(T)}E^{{h},{\eta,\xi}}[v^{-\theta/2}e^{\tilde{Z}_{T}}]\geq \inf_{h \in {\mathcal{H}}(T)}E^{{h},\eta,\xi}[v^{-\theta/2}e^{\tilde{Z}_{T}}]\geq  \tilde{w}(0,x).
\end{eqnarray}
Similarly, setting $t=0$ we get using conditions (5),(2) and (4) of the Proposition ,  we get the following upper value of the game, viz.\\
\begin{eqnarray}\label{3.6}
\inf_{h \in {\mathcal{H}}(T)}\sup_{(\eta,\xi) \in \mathcal{O}(T)}E^{{h},{\eta,\xi}}[v^{-\theta/2}e^{\tilde{Z}_{T}}]\leq
\sup_{(\eta,\xi) \in \mathcal{O}(T)}E^{{h},{\eta,\xi}}[v^{-\theta/2}e^{\tilde{Z}_{T}}]\leq \tilde{w}(0,x).
\end{eqnarray}
Also , setting $t=0$ and using conditions (5), (3) and (4) of the Proposition  we get,
\begin{eqnarray}\label{3.7}
E^{\hat{h},(\hat{\eta},\hat{\xi})}[\tilde{w}(T,X_{T})e^{\tilde{Z}_{T}}]&=& \tilde{w}(0,x)\no\\
&=&E^{\hat{h},(\hat{\eta},\hat{\xi})}[\exp\{\frac{\theta}{2}\int_{0}^{T}{g(X({s}),\hat{h}(X({s})),\hat{\eta}({s}),\hat{\xi}({s}),r(s);\theta)}ds\}v^{-\theta/2}].
\end{eqnarray}
From (\ref{3.5}), (\ref{3.6}) and (\ref{3.7}), and  that \\$\sup_{(\eta,\xi) \in \mathcal{O}(T)}\inf_{h \in \mathcal{H}(T)}[v^{-\theta/2}e^{\tilde{Z}_{T}}] \leq \inf_{h \in \mathcal{H}(T)} \sup_{(\eta,\xi) \in \mathcal{O}(T)}[v^{-\theta/2}e^{\tilde{Z}_{T}}]$ automatically holds,  the conclusion now follows.
\end{proof}
We now return to the game problem involving $u$ as the payoff function.

\begin{corollary}
$\underline{u}(0,x)=\bar{u}(0,x)=u(0,x)$
\end{corollary}
\begin{proof}
The value function $u$ and $\tilde{u}$ are related through the strictly monotone continuous transformation $\tilde{u}(t,x)=\exp(-\frac{\theta}{2}u(t,x))$. Thus admissible(Optimal) strategies for the exponentially transformed problem $\tilde{u}$ obtained via Proposition 3.1 are also admissible(optimal) for the  problem $u$. In other words,
\begin{eqnarray*}
u(0,x)& = & \sup_{h \in {\mathcal{H}}(T)}\inf_{(\eta,\xi) \in \mathcal{O}(T)}\{[{I}({v,x,{h},{\eta},\xi,0,T};\theta)]\},  \\
&=& \inf_{(\eta,\xi) \in \mathcal{O}(T)}\{\sup_{h \in {\mathcal{H}}(T)}[{I}({v,x,{h},{\eta},\xi,0,T};\theta)]\},  \\
&=& \inf_{(\eta,\xi) \in \mathcal{O}(T)}{I}({v,x,\hat{h},{\eta},\xi,0,T};\theta),  \\
&=& \sup_{h \in {\mathcal{H}}(T)}{I}({v,x,{h},\hat{\eta},\hat{\xi},0,T};\theta)={I}({v,x,\hat{h},\hat{\eta},\hat{\xi},0,T};\theta).
\end{eqnarray*}
Hence $\underline{u}(0,x)=\bar{u}(0,x)=u(0,x)$.
\end{proof}
\section{ Solving  game GII.}
{\bf \textit{\underline{{ Step 3 :}}}}\\
\indent We seek to find the game payoff function $u$ for the game  that would  satisfy all the conditions of our verification lemma given by Proposition 3.1 in terms of $u$.  The  Conditions (1)-(4) of verification lemma could be written in the compact form in terms of $u(t,x)$ as
\begin{eqnarray}\label{4.1}
{\mathcal{A}}^{\hat{h},\hat{\eta},\hat{\xi}}{u}(t,x)=0,\no\\
{u}(T,x)=\log{v}.
\end{eqnarray}
where the operator ${\mathcal{A}}^{{h},{\eta,\xi}}{u}(t,x)$  for any $h \in \mathbb{R}^{m}$ and $\eta \in \mathbb{R}^{n \times (n+m)}$, $\xi \in \mathbb{R}^{1 \times (n+m)}$ is given by,
\begin{eqnarray}\label{4.2}
{\mathcal{A}}^{{h},{\eta,\xi}}{u}(t,x)&=&\frac{\partial {u}(t,x)}{\partial t} +(b+Bx-\Lambda(\eta^{'}(s)X(s)+\xi^{'}(s)))-\frac{\theta}{2}\Lambda(\Sigma^{'}h))^{'}Du(t,x)+\frac{1}{2}tr(\Lambda\Lambda^{'}D^{2}u(t,x))\no\\
&-&\frac{\theta}{4}(Du(t,x))^{'}\Lambda\Lambda^{'}Du(t,x)-g(x,h,\eta,\xi,r;\theta).
\end{eqnarray}

The first order condition for $\hat{h}$ that maximizes $\mathcal{A}^{{h},\hat{\eta},\hat{\xi}}$ over all $\mathcal{H}(T)$ is given by,
\begin{eqnarray}\label{4.3}
\hat{h}({t})=\frac{2}{(\theta+2)}(\Sigma\Sigma^{'})^{-1}[\delta({t})-\Sigma(\hat{\eta}^{'}(t)X(t)+\hat{\xi}^{'})-\frac{\theta}{2}\Sigma\Lambda^{'}D u(t,x)].
\end{eqnarray}
Substituting (\ref{2.5}) in (\ref{4.2}) we obtain an  expression for the operator $\mathcal{A}^{h,\eta,\xi}$  in  ${\eta}^{'}(t)$ and ${\xi}^{'}(t)$. We minimize $\mathcal{A}^{h,\eta,\xi}$ over the set of controls $\mathcal{O}(T)$. As this operator is linear in ${\eta}^{'}(t)$ and ${\xi}^{'}(t)$,  we guess that the coefficient of the terms ${\eta}^{'}(t)$ and ${\xi}^{'}(t)$ vanish\footnote{Note that $\eta$ and $\xi$ are bounded, and the resulting conditions applying this guess must also gurantee that the factor process is indeed stable under the MMM. We show that in Remarks 2.1, 4.1 and  4.2.} leading to
\begin{eqnarray*}
\hat{h}(t)=-(\Sigma\Sigma^{'})^{-1}\Sigma\Lambda^{'}Du(t,x).
\end{eqnarray*}
Motivated by Kuroda and Nagai \cite{KN}, we will try the functional form for $u$ given by $u(t,x)=\frac{1}{2}x^{T}Q({t})x+ q^{T}({t})x + k({t})$ where $Q$ is an $n \times n$ symmetric matrix, $q$ is a n-element column vector and $k$ is a scalar. Hence
\begin{eqnarray}\label{4.4}
\hat{h}(t)=-(\Sigma\Sigma^{'})^{-1}\Sigma\Lambda^{'}(Q(t)X(t)+q(t)).
\end{eqnarray}
This when substituted in (\ref{4.3}) yields,
\begin{eqnarray}\label{4.5}
-\Sigma\Lambda^{'}(Q(t)X(t)+q(t))=\delta(t)-\Sigma(\hat{\eta}^{'}(t)X(t)+\hat{\xi}^{'}(t)).
\end{eqnarray}
which further yields,
\begin{equation}\label{4.6}
 \left.\begin{aligned}
        \hat{\eta}(t)=(Q^{'}(t)\Lambda\Sigma^{'}+A^{'})(\Sigma\Sigma^{'})^{-1}\Sigma ,\\
        \hat{\xi}(t)=\bigg((a-r(t)1)^{'}+q^{'}(t)\Lambda\Sigma^{'}\bigg)(\Sigma\Sigma^{'})^{-1}\Sigma.
       \end{aligned}
 \right\}
\end{equation}
Thus $\hat{h}$ is a local maximizing control and  ($\hat{\eta},\hat{\xi}$) is  a local minimizer control that constitutes the saddle-point equilibrium for game ({\bf GII}).\\
\begin{rem}
From Remark 2.1 and equation (\ref{4.6}), it can be seen that $\mathbb{P}^{\hat{\eta},\hat{\xi}}$ is a minimal martingale measure provided $Q^{'}(t)\Lambda\Sigma^{'}(\Sigma\Sigma^{'})^{-1}\Sigma=0$ and $q^{'}(t)\Lambda\Sigma^{'}(\Sigma\Sigma^{'})^{-1}\Sigma=0$ for $t \leq T$.
\end{rem}
\begin{rem}
From Remark 4.1, and equation (\ref{4.4}) it is clear that if the game equilibrium measure corresponds to the minimal martingale measure then the  optimal investor strategy satisfies $\hat{h}^{'}(t)\Sigma=0$. Hence if the portfolio model does not permit short selling then the optimal investor strategy at game equilibrium is the no-regret strategy i.e ($\hat{h}(t)$=0).
\end{rem}
\begin{rem}
In the case where  the factor process and the security(stock) price process has independent noise i.e $\Sigma \Lambda^{'}$=0 , then from Remarks 4.1-4.2, it is obvious that at optimality, the worst case strategy is the no-regret strategy and the factor process always evolves under the minimal martingale measure since the game equilibrium measure is the minimal martingale measure.
\end{rem}
 \indent As like in Kuroda and Nagai \cite{KN}, we can verify that  $u(t,x)=\frac{1}{2}x^{'}Q(t)x+q^{'}(t)x+k(t)$ satisfies the  HJB PDE  i.e conditions (1)-(4) of the Proposition 3.1 provided\\
$\bullet$~~an $n \times n$ symmetric non-negative matrix $Q$ satisfies the following matrix Riccati equation given as
\begin{eqnarray}\label{4.7}
\frac{dQ(t)}{dt}-Q(t)K_{0}Q(t)+K_{1}^{'}Q(t)+Q(t)K_{1}=0~~~0 \leq t \leq T,~~Q(T)=0.
\end{eqnarray}
{where}
\begin{eqnarray*}
K_{0}&=&\frac{\theta}{2}\Lambda\bigg(I-\frac{\theta-2}{\theta}\Sigma^{'}(\Sigma\Sigma^{'})^{-1}\Sigma\bigg)\Lambda^{'},\\
K_{1}&=&B-\Lambda\eta^{'}(t)-\Lambda\Sigma^{'}{(\Sigma\Sigma^{'})}^{-1}A+\Lambda \Sigma^{'}{(\Sigma\Sigma^{'})}^{-1}\Sigma \eta^{'}.\\
\end{eqnarray*}
$\bullet$~~The $n$ element column vector $q(t)$satisfies the following linear ordinary differential equation \\
for $0 \leq t \leq T$.
\begin{eqnarray}\label{4.8}
\frac{d q(t)}{dt}&+&(K_{1}^{'}-Q(t)K_{0})q(t)+Q(t)b-Q^{'}(t)\Lambda \Sigma^{'}(\Sigma\Sigma^{'})^{-1}(a-r(t)1)\no\\
&+&Q^{'}(t)\Lambda\Sigma^{'}(\Sigma\Sigma^{'})^{-1}\Sigma\xi^{'}(t)-Q^{'}(t)\Lambda\xi^{'}(t)=0,\no\\
q(T)&=&0.
\end{eqnarray}
$\bullet$ and the constant $k(t)$  is a solution to
\begin{eqnarray}\label{4.9}
\frac{dk(t)}{dt}&+& b^{'}q(t)+\frac{\theta-2}{4}q^{'}(t)\Lambda\Sigma^{'}(\Sigma\Sigma^{'})^{-1}\Sigma\Lambda^{'}q(t)\no\\
&+&r-q^{'}(t)\Lambda\Sigma^{'}(\Sigma\Sigma^{'})^{-1}(a-r(t)1)+q^{'}(t)\Lambda\Sigma^{'}(\Sigma\Sigma^{'})^{-1}\Sigma\xi^{'}(t)\no\\
&-&\xi(t)\Lambda^{'}q(t)+\frac{2-\theta}{4}q^{'}(t)\Lambda\Lambda^{'}q(t)=0\no\\
, \forall 0 \leq t \leq T, \no\\
k(T)&=&\log(v).
\end{eqnarray}
The fourth condition of Proposition 3.1 is  obvious from the terminal conditions of $Q$, $q$ and $k$. To  show that condition (5) of Proposition 3.1 is satisfied by the choice of our payoff function, we need to show that $E^{h,(\eta,\xi)}({<D\tilde{u}~\Lambda e^{Z}, D\tilde{u}~\Lambda e^{Z}>}_{t})< \infty $ $\forall t \in [0,T]$ where $<\cdot,\cdot>$ as usual symbolizes quadratic co-variation.  To show this we argue as follows. Processes $Q\triangleq (Q(t))_{0 \leq t \leq T}$ and $q \triangleq (q(t))_{0 \leq t \leq T}$ are bounded since they are continuous on the compact support $[0,T]$. By standard existence-uniqueness argument for stochastic differential equation (refer Gihman and Skorokhod \cite{GS}), $X \in$ $ L^{2}(\mathbb{P}^{h,(\eta,\xi)})$. Since $D \tilde{u}$ is linear in $X$ with controls ($\eta,\xi$) assumed bounded, we also have that $D\tilde{u} \in$ $ L^{2}(\mathbb{P}^{h,\eta,\xi})$. To complete the argument it remains to be shown that $\tilde{u}$ is bounded which we show now.
\begin{lem}
$0< \tilde{u} < \exp(-\frac{\theta}{2}\int_{0}^{T-t}{r{(s+t)}}ds)v^{-\theta/2}$.
\end{lem}
\begin{proof}
From the definition of $\tilde{u}$ in (\ref{2.11}), for any optimal control $\mathcal{O}(T)$, the strategy $\hat{h}(t)= 0$ for $t \leq T$ is sub-optimal, and hence will provide an upper bound on $\tilde{u}$.  Hence from the definition of $g$ in equation (\ref{2.5}) to obtain the upper bound. Formally we write these statements as,
 \begin{eqnarray*}
\tilde{u}(t,x)&=&\inf_{h \in {\mathcal{H}}(T)}E^{h,\hat{\eta},\hat{\xi}}[\exp\{\frac{\theta}{2}\int_{0}^{T-t}{g(X({s}),h({s}),\hat{\eta}({s}),\hat{\xi}({s}),r(s+t);\theta)}ds\}v^{-\theta/2}],\\
&\leq & E^{0,\hat{\eta},\hat{\xi}}[\exp\{\frac{\theta}{2}\int_{0}^{T-t}{g(X({s}),0,\hat{\eta}({s}),\hat{\xi}({s}),r(s+t);\theta)}ds\}v^{-\theta/2}],\\
&= & \exp(-\frac{\theta}{2}\int_{0}^{T-t}{r({s+t})}ds)v^{-\theta/2}.
 \end{eqnarray*}
 Hence the conclusion follows.
\end{proof}
We now formalize the solution to this game ({\bf GI}).\\
{\bf \textit{ \underline{Step 4:}}}\\
We first show that the  controls belonging to $ \mathcal{H}(T)$ and $\mathcal{O}(T)$ satisfy the following change of measure criterion.\\
\begin{lem}
From  the choice of space of controls  $h \in \mathcal{H}(T)$ and $(\eta,\xi) \in \mathcal{O}(T)$, we have
\begin{eqnarray}\label{4.10}
E[\mathcal{E}\bigg(-\frac{\theta}{2}\int_{0}{[(Q({t})X({t})+q({t}))\Lambda+{h}^{'}({t})\Sigma]dW^{\eta,\xi}({t})}\bigg)_{T}]=1.
\end{eqnarray}
\end{lem}
\begin{proof}
Above result holds if the following Kazamaki condition, \\ $E[\exp(\int_{0}^{t}{\theta(\frac{(Q({s})X({s})+q({s}))\Lambda+{h}^{'}({s})\Sigma}{2})}dW^{\eta,\xi}({s}))]<\infty$ $\forall~t \in [0,T]$ is satisfied. By an application of Cauchy-Schwartz inequality we have $\forall~ t \in [0,T]$ ,
\begin{eqnarray*}
&&E[\exp(\int_{0}^{t}{\theta(\frac{(Q({s})X({s})+q({s}))\Lambda+({h}^{'}({s})\Sigma)}{2})}dW^{\eta,\xi}({s}))] \no\\ 
&\leq&(E[e^{\int_{0}^{t}{\theta{(Q({s})X({s})+q({s}))\Lambda}}dW^{\eta,\xi}({s})}])^{1/2} \times {(E[e^{\int_{0}^{t}{\theta{{({h}^{'}({s})\Sigma)}}}dW^{\eta,\xi}({s})}])}^{1/2}
\end{eqnarray*}
Since $X$ is Gaussian process, mimicking arguments similar to Lemma 2.1, we have that \\ $(E[e^{\int_{0}^{t}{\theta{(Q({s})X({s})+q({s}))\Lambda}}dW^{\eta,\xi}_{s}}])^{1/2}<\infty $ $\forall~t \in [0,T]$. From assumption on the space of controls $\mathcal{H}(T)$, one can conclude that ${(E[e^{\int_{0}^{t}{\theta{{({h}^{'}({s})\Sigma)}}}dW^{\eta,\xi}({s})}])}^{1/2}<\infty$ for $ t \in [0,T]$. Hence the Kazamaki condition holds true and the conclusion follows.
\end{proof}
\indent We now show that the saddle-point equilibrium controls obtained by solving game ({\bf GII}) is  in fact also a saddle-point equilibrium for the original game problem ({\bf GI}).\\
\begin{prop}
If there exist a solution $Q$ to the matrix Ricatti equation (\ref{4.7}) , then the  saddle point equilibrium strategies $\hat{h}$ and $(\hat{\eta},\hat{\xi})$ obtained from (\ref{4.4}) and (\ref{4.6}) respectively as a result of solving the auxiliary game ({\bf GII}) where $q$  is a solution to (\ref{4.8}) and and $k$ is a solution of (\ref{4.9}) is in fact also the saddle-point equilibrium for the finite horizon game ({\bf GI}), namely,
\begin{eqnarray*}
 \sup_{h \in \hat{\mathcal{H}}(T)}\inf_{(\eta,\xi)\in \hat{\mathcal{O}}(T)}\tilde{J}(v,h,\eta,\xi,T;\theta)&=& \inf_{(\eta,\xi)\in \hat{\mathcal{O}}(T)}\sup_{h \in \hat{\mathcal{H}}(T)}\tilde{J}(v,h,\eta,\xi,T;\theta),\\
&=&\tilde{J}(v,\hat{h},\hat{\eta},\hat{\xi},T;\theta,)\\
&=&\frac{1}{2}x^{'}Q(0)x+q^{'}(0)x+k(0).
 \end{eqnarray*}
where,
\begin{eqnarray*}
\tilde{J}(v,h,\eta,\xi,T;\theta)\triangleq \frac{-2}{\theta}\log E^{\eta,\xi}[\exp{[\frac{-\theta}{2}\log{V^{h,\eta,\xi}({T})}]}].
\end{eqnarray*}
 \end{prop}
\begin{proof}
Define,
\begin{eqnarray}\label{4.11}
\bar{Z}_{s}&=&\bar{Z}_{s}(h,\eta,\xi)=\frac{\theta}{2}\bigg\{\int_{0}^{s}{g(X({\tau}),h({\tau}),\eta({\tau}),\xi({\tau}),r({t+\tau});\theta)}d\tau-{(h^{'}({\tau})\Sigma)}dW^{\eta,\xi}({\tau})\no\\
&-&\frac{\theta}{4}{(h^{'}({\tau})\Sigma)}^{'}{(h^{'}({\tau})\Sigma)}d\tau\bigg\}.
\end{eqnarray}
Also define, $\chi(t,x)=-\frac{\theta}{2}(u(t,x)-\log{v})$.
From some straightforward calculations provided in the Appendix  we obtain the following relation,
\begin{eqnarray}\label{4.12}
&&\exp\{\chi(T,X(T-t))+\bar{Z}(T-t)\}=\exp(\chi(t,x))\exp \bigg[\int_{0}^{T-t}{-\frac{\theta}{2}(\mathcal{A}^{h,\eta,\xi} u(t+s,X_s))}ds\no\\
&-&\int_{0}^{T-t}{\frac{\theta}{2}[Du(t+s,X_s)^{'}\Lambda+(h^{'}(t)\Sigma)]}dW^{\eta,\xi}_t\no\\
&-&\int_{0}^{T-t}{\frac{{\theta}^{2}}{8}{[Du(t+s,X_s)^{'}+(h^{'}(t)\Sigma)][Du(t+s,X_s)^{'}+h^{'}(t)\Sigma]^{'}}ds}\bigg].\no\\
\end{eqnarray}
We have shown that  the saddle-point equilibrium strategies $\hat{h}$ and $(\hat{\eta},\hat{\xi})$ deduced by solving  game ({\bf GI}) with corresponding game payoff function $u$ satisfies  conditions (1)-(5) of Proposition 3.1. Therefore from condition(4) of Proposition 3.1,  we have $\chi(T,x)=0$. Moreover ${({V^{h,\eta,\xi}({T})})}^{-\theta/2}=v^{-\theta/2}e^{\bar{Z}_{T}}.$ Setting $t=0$ and taking condition (1) of  Proposition 3.1 into account for $\eta=\hat{\eta}, \xi=\hat{\xi}$, and for any $h \in \hat {\mathcal{H}}(T)$ we see from (\ref{4.11}) that \\
\begin{eqnarray*}
{({V^{h,\eta,\xi}({T})})}^{-\theta/2}&\geq& e^{-\frac{\theta}{2}u(0,x)}\exp\bigg[-\int_{0}^{T}{\frac{\theta}{2}[Du(s,X(s))^{'}\Lambda+h^{'}({s})\Sigma]}dW^{\eta,\xi}(s)\no\\
&-&\int_{0}^{T}{\frac{{\theta}^{2}}{8}{[Du(s,X(s))^{'}+h^{'}(s) \Sigma][Du(s,X(s))^{'}+h^{'}(s)\Sigma]^{'}}ds}\bigg].
\end{eqnarray*}
Now by taking expectations w.r.t to the physical probability measure $\mathbb{P}^{\eta,\xi}$ on both sides of above equation and using Lemma 4.5, we obtain\\
\begin{eqnarray*}
\tilde{J}(v,h,\eta,\xi,T;\theta) \leq u(0,x).
\end{eqnarray*}
This inequality is true for all $h \in {\mathcal{H}}(T)$. Hence we have,
\begin{eqnarray*}
\sup_{h \in {\mathcal{H}}(T)} \tilde{J}(v,h,\eta,\xi,T;\theta)  \leq u(0,x).
\end{eqnarray*}
Hence we have,
\begin{eqnarray}\label{4.13}
\inf_{(\eta,\xi) \in \mathcal{O}(T)}\sup_{h \in {\mathcal{H}}(T)} \tilde{J}(v,h,\eta,\xi,T;\theta) \leq \sup_{h \in {\mathcal{H}}(T)} \tilde{J}(v,h,\eta,\xi,T;\theta)  \leq u(0,x).
\end{eqnarray}
Likewise, setting $t=0$ and taking condition  (2) and condition (5) of Proposition 3.1 into account  we see that
\begin{eqnarray}\label{4.14}
\sup_{h \in {\mathcal{H}}(T)} \inf_{(\eta,\xi) \in \mathcal{O}(T)} \tilde{J}(v,h,\eta,\xi,T;\theta)\geq u(0,x)\geq \inf_{(\eta, \xi) \in \mathcal{O}(T)}\sup_{h \in {\mathcal{H}}(T)} \tilde{J}(v,h,\eta,\xi,T;\theta).
\end{eqnarray}
Similarly ,setting $t=0$ and taking condition  (3) and (5) of Proposition 3.1 into account for $h=\hat{h},\gamma=\hat{\gamma}$ such that $\hat{h} \in \mathcal{H}(T)$ and $(\hat{\eta},\hat{\xi}) \in \mathcal{O}(T)$ we see that
\begin{eqnarray}\label{4.15}
\tilde{J}(v,\hat{h},\hat{\eta},\hat{\xi},T;\theta)= u(0,x).
\end{eqnarray}
From (\ref{4.13})-(\ref{4.15}) and the fact that \\ $\sup_{h \in {\mathcal{H}}(T)} \inf_{(\eta,\xi) \in \mathcal{O}(T)} \tilde{J}(v,h,\eta,\xi,T;\theta) \leq  \inf_{(\eta, \xi) \in \mathcal{O}(T)}\sup_{h \in {\mathcal{H}}(T)} \tilde{J}(v,h,\eta,\xi,T;\theta)$ is automatically true,  we conclude that the saddle-point equilibrium controls obtained by solving game ({\bf GII})  in fact also constitutes saddle-point strategy for the original game ({\bf GI}).
\end{proof}

\appendix

\section{}
\textit{As part of the proof of Proposition 4.1}\\
Let $\chi(t,x)=-\frac{\theta}{2}(u(t,x)-\log{v})$ and $Lu(t,x)=\frac{1}{2}tr(\Lambda\Lambda^{'}D^{2}u(t,x))+(b+Bx-\Lambda(\eta^{'}x+\xi^{'}))^{'}Du(t,x)$\\
Hence, we have\\
\begin{eqnarray*}
d\chi(t+s,X(s))&=&-\frac{\theta}{2}(\frac{\partial u}{\partial t}+Lu)(t+s,X(s))ds-\frac{\theta}{2}Du(t+s,X(s))^{'}\Lambda dW^{\eta,\xi}(s)\\
\therefore  \frac{d\exp\{\chi(t+s,X(s))\}}{\exp\{\chi(t+s,X(s))\}}&=& -\frac{\theta}{2}(\frac{\partial u}{\partial t}(t,x)+{L}u)(t+s,X(s))-\frac{\theta}{2}Du(t+s,X(s))^{'}\Lambda dW^{\eta,\xi}(s)\\
&+&\frac{\theta^{2}}{8}Du^{'}\Lambda\Lambda^{'}Du(t+s,X(s))ds
\end{eqnarray*}
\begin{eqnarray*}
\therefore \frac{d\exp\{\chi(t+s,X(s))\}\exp\{Z(s)\}}{\exp\{\chi(t+s,X(s))\}\exp\{Z(s)\}}&=& -\frac{\theta}{2}(\frac{\partial u}{\partial t}(t,x)+{L}u)(t+s,X(s))
-\frac{\theta}{2}{Du(t+s,X(s))}^{'}\Lambda dW^{\eta,\xi}(s)\\
&+&\frac{\theta^{2}}{8}Du^{'}\Lambda\Lambda^{'}Du(t+s,X(s))ds +\frac{\theta}{2}{g(X(t),h(t),\eta(t),\xi(t),r(s+t);\theta)}ds\\
&-&\frac{\theta}{2}h^{'}(s)\Sigma dW^{\eta,\xi}(s)
+\frac{\theta^{2}}{4}h^{'}(s)\Sigma\Lambda^{'}Du(t+s,X(s))ds
\end{eqnarray*}
Integrating above equation we yield (\ref{4.12}).



\acks
The financial support of Chancellor's Scholarship of University of Warwick is gratefully acknowledged. This article is dedicated as an ode to the rich legacy of research work of Wendell H. Fleming \cite{WHF} and Peter Whittle \cite{W}.


%
%
%
%

\end{document}